\documentclass[11pt]{article}
\usepackage{amssymb}
\usepackage{amsthm}
\usepackage{latexsym}
\usepackage{makecell}
\usepackage{tikz}
\usepackage{amsmath}
\usepackage{mathabx}
\usepackage[all]{xy}
\usepackage{titling}
\thanksmarkseries{arabic}
\usepackage{hyperref}
\usepackage{accents}
\hypersetup{
  colorlinks   = true, %Colours links instead of ugly boxes
  urlcolor     = blue, %Colour for external hyperlinks
  linkcolor    = blue, %Colour of internal links
  citecolor   = red %Colour of citations
}
\theoremstyle{definition}
\newtheorem{theo}{Theorem}[section]
\newtheorem{prop}[theo]{Proposition}
\newtheorem{cor}[theo]{Corollary}
\newtheorem{lem}[theo]{Lemma}
\newtheorem{notation}[theo]{Notation}
\newtheorem{defi}[theo]{Definition}

\newtheorem{defitheo}[theo]{Definition/Theorem}
\newtheorem{exa}[theo]{Example}
\newtheorem{rem}[theo]{Remark}

\newtheorem{problem}[theo]{Problem}
\numberwithin{equation}{section}
\newcommand*\circled[1]{\tikz[baseline=(char.base)]{
            \node[shape=circle,draw,inner sep=2pt] (char) {#1};}}
\newcommand{\ev}{\mathrm{ev}}
\newcommand{\N}{{\mathbb N}}
\newcommand{\F}{{\mathbb F}}

\newcommand{\Z}{{\mathbb Z}}

\newcommand{\C}{{\mathbb C}}
\newcommand{\R}{{\mathbb R}}

\newcommand{\cC}{{\mathcal C}}

\newcommand{\cG}{{\mathcal G}}

\newcommand{\cM}{{\mathcal M}}
\newcommand{\cP}{{\mathcal P}}
\newcommand{\cO}{{\mathcal O}}
\newcommand{\cQ}{{\mathcal Q}}
\newcommand{\cB}{{\mathcal B}}

\newcommand{\cN}{{\mathcal N}}
\newcommand{\cU}{{\mathcal U}}

\newcommand{\mm}{{\mathfrak m}}
\newcommand{\Aut}{\mbox{\rm Aut}}
\newcommand{\rMon}{\mbox{\rm rMon}}
\newcommand{\Iso}{\mbox{\rm Iso}}
\newcommand{\bidual}[1]{{\widehat{\phantom{\big|}\hspace*{.5em}}\hspace*{-.9em}\widehat{#1}}}

\newcommand{\Char}{\mbox{${\rm char}$}}
\newcommand{\hwt}{\mbox{${\rm wt}_{\rm H}$}}

\newcommand{\ann}{\mbox{\rm ann}}
\newcommand{\soc}{\mbox{\rm soc}}
\newcommand{\rad}{\mbox{\rm rad}}
\newcommand{\im}{\mbox{\rm im}\,}

\newcommand{\Rhat}{\mbox{$\widehat{R}$}}
\newcommand{\T}{\mbox{$\!^{\sf T}$}}
\newcommand{\spann}{\mbox{\rm span}\,}

\newcommand{\swt}{\mbox{${\rm wt}_{\rm s}$}}

\newcommand{\sbt}{\raisebox{.2ex}{\mbox{$\scriptscriptstyle\bullet\,$}}}
\newcommand{\Mon}{\mbox{\rm Mon}}
\newcommand{\group}[1]{\mbox{$\langle{#1}\rangle$}}
\newcommand{\inner}[2]{\mbox{$\langle\,{#1}\,|\,{#2}\,\rangle$}}
\newcommand{\inners}[2]{\mbox{$\langle{\,{#1}\,}|\,{#2}\,\rangle_{\rm s}$}}

\newcommand{\ov}[1]{\mbox{$\overline{#1}$}}
\newcommand{\GL}{\mathrm{GL}}
\newcommand{\SL}{\mathrm{SL}}
\newcommand{\Symp}{\mathrm{Symp}}

\newcommand{\diag}{\textup{diag}}

\newcommand{\dist}{\textup{dist}}
\newcommand{\cBn}{\mbox{$\cB^{\otimes n}$}}

\newcounter{alp}
\newcounter{ara}
\newcounter{rom}

\newenvironment{arabiclist}{\begin{list}{(\arabic{ara})\hfill}{\usecounter{ara}
     \topsep0.4ex \labelwidth.6cm \leftmargin.6cm \labelsep0cm
     \rightmargin0cm \parsep0ex \itemsep0ex}}{\end{list}}
\begin{document}
\title{Symplectic Isometries of Stabilizer Codes}
\date{\today}
\author{Tefjol Pllaha\\University of Kentucky\\Department of Mathematics\\715 Paterson Office Tower\\ Lexington KY 40506-0027, USA\\ tefjol.pllaha@uky.edu
}

\maketitle

{\bf Abstract:} In this paper we study the equivalence of quantum stabilizer codes via symplectic isometries of stabilizer codes. We define monomially and symplectically equivalent stabilizer codes and determine how different the two notions can be. Further, we show that to monomial maps correspond local Clifford operators. We relate the latter with the LU-LC Conjecture.



%\mbox{}\hfill\today

\section{Introduction}
Quantum stabilizer codes constitute the most important class of quantum error-correcting codes due to their rich structure and strong connections with classical codes. It is well-known \cite{CRSS98} that to a quantum stabilizer code corresponds a self-orthogonal (with respect to a certain symplectic bilinear form) classical code. We refer to the latter as a stabilizer code. Subsequently the study was extended to finite fields \cite{AK01,KKKS06}. As it turned out, the crucial property of finite fields that allowed the generalization was the existence of a generating character. This led Nadella and Klappenecker \cite{NK12} to the study of quantum stabilizer codes over Frobenius rings, where they show the existence along with structural results when restricted to chain rings. While they show that quantum stabilizer codes over Frobenius rings cannot outperform quantum stabilizer codes over fields, they point out the simpler arithmetic of the former. In \cite{me2} the authors generalize the results to the much larger class of local Frobenius rings. 

In this paper we extend the study of \cite[Section 7]{me2} and focus on the equivalence of quantum stabilizer codes. In particular, we discuss the structure of symplectic isometries of stabilizer codes. A symplectic isometry is a linear map that preserves the symplectic bilinear form and the symplectic weight. A particularly nice class of symplectic isometries are the so-called monomial maps. However, as shown in \cite[Ex. 7.3]{me2}, there exist symplectic isometries that are not monomial maps. This led to the problem \cite[Q. 7.4]{me2} of determining how far from being a monomial map a symplectic isometry is. This type of question is well-studied in classical coding theory, and it is commonly referred as MacWilliams Extension Theorem. F. J. MacWilliams showed in her PhD thesis that every Hamming isometry between linear codes over fields is a monomial map. Thus, in that particular case, there is no difference between monomial maps and Hamming isometries. The result has been vastly generalized by considering different weight functions and different alphabets; see \cite{W99, GS00, GNW04, W09,me1} and the references therein. On the other hand, MacWilliams Extension Theorem does not hold for additive codes endowed with the Hamming weight. This led Wood \cite{Wood16} to the study of isometry groups of additive codes.

The symplectic weight in $R^{2n}$ is simply the Hamming weight in $(R^2)^n$ up to a change of coordinates that we call $\gamma$. This elementary observation allows us to make use of the work of Wood \cite{Wood16}. However, to study isometry groups of stabilizer codes we need to take care of self-orthogonality. To this end, we modify the notions of \cite{Wood16} and associate to a stabilizer $\cC$ code two isometry groups: $\Symp(\cC)$ and $\Mon_{\text{SL}}(\cC)$, which satisfy $\Mon_{\text{SL}}(\cC)\subsetneq \Symp(\cC)$. The first task of this paper is to establish how different can these two groups be. We show that the difference can be as big as possible when $R=\F_q$ is a finite field and give partial results when $R$ is local Frobenius ring. Secondly, we consider the equivalence and local equivalence of quantum stabilizer codes. We show that monomial maps correspond to local Clifford operators. 

The paper is organized as follows. In Section \ref{S-2} we provide some background on finite (commutative) Frobenius rings. Section \ref{S-3} draws the connections between quantum stabilizer codes and stabilizer codes in the general setting of Frobenius rings. In Section \ref{Sec-Symp} we study in details symplectic isometries of stabilizer codes. First we view the isometry groups as groups of matrices. This gives a unified approach that takes care of the change of coordinate $\gamma$ between $R^{2n}$ and $(R^2)^n$. Secondly, we show that concatenation preserves the isometry groups. Subsequently, we make use of the latter and \cite[Thm. 5.1]{Wood16} to produce stabilizer codes with predetermined isometry groups. In Section \ref{Sec-LULC} we apply the results of Section \ref{Sec-Symp} to the so-called LU-LC conjecture. As mentioned, monomial maps completely determine local Clifford equivalence. On the other hand it is not clear how general symplectic isometries relate to local unitary equivalence. Understanding the latter yields a systematic way of creating examples that disprove the LU-LC Conjecture. Finally, we end the paper with some conclusions and directions for future research.

%%%%%%%%%%%%%%%%%%%%%%%%%%%%%%%%%%%%%%%%%%%%%%%%%%%%%%%%%%%
\section{Frobenius Rings}\label{S-2}
In this section we collect a few facts about Frobenius rings. Let~$A$ be a finite abelian group.
Its \textbf{character group} is defined as the set $\widehat{A} := \text{Hom}(A, \mathbb{C}^*)$ of all group homomorphisms from~$(A,+)$ to~$\C^*$,
endowed with addition $(\chi_1 + \chi_2)(a) = \chi_1(a)\chi_2(a)$ for all $\chi_i\in\widehat{A}$ and $a\in A$.
Then $\widehat{A}$ is again an abelian group.
Its zero element is $\varepsilon_A \in \widehat{A}$ given by $\varepsilon_A(a) =1$ for all $a\in A$.
Elements of $\widehat{A}$ are called \textbf{characters} and $\varepsilon_A$ is the \textbf{principal character} of $A$. The additive inverse of $\chi \in \widehat{A}$ is given by $(-\chi)(a) := \ov{\chi(a)}$, where $\ov{\sbt}$ denotes the complex conjugate. It is well-known that $A\cong \widehat{A}$ as groups, though the isomorphism is not natural. We have a natural isomorphism of groups
\begin{equation}\label{e-bidualA}
  \zeta_A: A\longmapsto \bidual{A}, \,\,\,\\
   a\longmapsto \left\{\!\!\begin{array}{rcl}\ev_a:\widehat{A}&\longrightarrow&\C^*\\ \chi&\longmapsto&\chi(a)\end{array}\right.,
\end{equation}
and thus we identify $A$ and $\bidual{A}$. The kernel of a character $\chi\in \widehat{A}$ is $\ker \chi:= \{a\in A \mid \chi(a) = 1\}$ and it is well-known that 
\begin{equation}\label{e-ker0}
\bigcap_{\chi \in \widehat{A}} \ker \chi = \{0\}.
\end{equation}
We will focus on the additive group of finite commutative rings. To this end, let $R$ be a finite commutative ring with identity, and consider the character group $\widehat{R}$. As mentioned we have $R\cong \widehat{R}$ as groups. Moreover, in this case, the character group~$\Rhat$ carries an $R$-module structure via scalar multiplication
\begin{equation}\label{e-bimodule}
    (r\!\cdot\!\chi)(v)=\chi(vr) \text{ for all } r\in R\text{ and }v\in R.
\end{equation}

%%%%%%%%%%%%%%%%%%%%%%%%%%%%
\begin{defi}\label{D-Frob}
A finite commutative ring~$R$ is called \textbf{Frobenius} if $\Rhat\cong R$ as $R$-modules.
\end{defi}
%%%%%%%%%%%%%%%%%%%%%%%%%%%%%%%%%%%%%
It is an immediate consequence of the definition that in this case there exists a character~$\chi$ such that $\Rhat= R\!\cdot\!\chi$.
Any such character is called a \textbf{generating character} of~$R$.
With the aid of \eqref{e-ker0} the reader will verify that a character $\chi \in \widehat{R}$ is generating iff $\ker \chi$ contains no non-zero ideals. It follows from this equivalence that any two generating characters~$\chi,\,\chi'$ differ by a unit, i.e., $\chi'=u\!\cdot\!\chi$ for some $u\in R^*$.

Frobenius rings have been historically defined via the socle $\soc(R)$ and the Jacobson radical $\rad(R)$; see Theorem \ref{T-Frob} below. This character-theoretic approach has been exploited in detail in \cite{CG92, W99, HOnold01}. Frobenius rings have been characterized by Wood \cite[Thm 6.3, Thm. 6.4]{W99} as those commutative rings that satisfy \textbf{MacWilliams Extension Theorem} \cite{MacWilliams62} for the Hamming weight. Classical coding theory over finite Frobenius rings is a well established area. This paper, along with \cite{NK12,me2}, provide yet another evidence of the importance of Frobenius rings in quantum error-correction. 

\begin{theo}\label{T-Frob}
Let $R$ be a finite commutative ring. Then the following are equivalent.
\begin{arabiclist}
\item $R$ is Frobenius.
\item $\soc(R) \cong R/\rad(R)$ as $R$-modules.
\item There exists $\alpha \in R$ such that $\soc(R) = \alpha R$.
\end{arabiclist}
\end{theo}

For a Frobenius ring, the socle and the Jacobson radical are very closely related. Namely, let $\ann(I)$ the annihilator of the ideal $I\subseteq R$. Then 
\begin{equation}\label{e-RadSoc}
\ann(\rad(R)) = \soc(R) \text{ and } \ann(\soc(R)) = \rad(R);
\end{equation}
see \cite[Cor. 15.7]{lam-lmr} for instance. 

\begin{rem}\label{R-LFR}
Let $R$ be a local Frobenius ring with unique maximal ideal $\mm$. Denote $R/\mm:=\F_q$ the residue field. For $r\in R$ we will denote $\ov{r}:=r+\mm \in R/\mm$. Then of course $\rad(R) = \mm$ and $\soc(R) = \alpha R$ for some $\alpha \in R$ as in Theorem \ref{T-Frob}(3). In this case \eqref{e-RadSoc} reduces to 
\begin{equation}\label{e-ann}
\ann(\mm) = \alpha R \text{ and } \ann(\alpha R) = \mm.
\end{equation}
Thanks to \eqref{e-ann} we obtain a well-defined isomorphism $\rho : \alpha R \longrightarrow \F_q$ via $\alpha r \longmapsto \ov{r}$. Moreover, we obtain a well-defined multiplication 
\begin{equation}\label{e-FqVS}
\ov{r}\cdot x = rx, \text{ for all } \ov{r}\in \F_q \text{ and } x \in \alpha R,
\end{equation}
which makes $\alpha R$ a $\F_q$-vector space. In particular, for any $n\in \N$, $\F_q$-linear maps and $R$-linear maps of $(\alpha R)^n$ coincide.
\end{rem}

\section{Stabilizer Codes}\label{S-3}
In this section we define stabilizer codes over  Frobenius rings and motivate the definitions by drawing connections with quantum error-correction and quantum stabilizer codes. The approach was first studied in \cite{NK12} where the authors generalize the definitions of non-binary quantum stabilizer codes \cite{AK01,KKKS06} by making use of the existence of a generating character.

Let $R$ be a finite Frobenius ring with cardinality $|R| = d$ and generating character $\chi$. Fix an orthonormal basis $\cB = \{v_x\mid x\in R\}$ of $\C^d$ indexed by the ring elements. The pair $(\C^d, \cB)$ is called a \textbf{qudit}. For $a\in R$ define the two linear maps $X(a),\,Z(a):\C^{q}\longrightarrow\C^{q}$ where their action on the
basis~$\cB$ given by
\begin{equation}\label{e-XZ}
  X(a)(v_x)=v_{x+a},\quad Z(a)(v_x)=\chi(ax)v_x \text{ for all }x\in R.
\end{equation}
A \textbf{$n$-qudit} is the pair $(\C^{d^n}, \cB^{\otimes n})$ where
\begin{equation}\label{e-Bn}
  \cBn=\{v_x=v_{x_1}\otimes\ldots\otimes v_{x_n}\mid x=(x_1,\ldots,x_n)\in R^n\},
\end{equation}
and we identify $\C^{d^n} \cong (\C^d)^{\otimes n}$. For $a=(a_1,\ldots,a_n)\in R^n$ set
\begin{equation}\label{e-XZtensor}
  X(a)=X(a_1)\otimes\ldots\otimes X(a_n),\quad Z(a)=Z(a_1)\otimes\ldots\otimes Z(a_n).
\end{equation}
If we use the standard dot product in $R^n$, that is, $ax=a\cdot x=\sum_{i=1}^n a_i x_i$, then \eqref{e-XZtensor} reads as
\begin{equation}\label{e-XZ}
  X(a)(v_x)=v_{x+a}\ \text{ and }
  Z(a)(v_x)=\chi(a x)v_x \ \text{ for all }a,\,x\in R^n.
\end{equation}
Using properties of characters it is easy to see that $X(a),\,Z(a)$ are unitary maps for all $a\in R^n$.
%%%%%%%%%%%%%%%%%%%%%%%%%%%%%%
\begin{lem}[\mbox{\cite[Prop.~4 and proof]{NK12}}]\label{L-MultEn}
Let $(a,b),(a',b')\in R^{2n}$ and consider the unitary maps $P=X(a)Z(b),\,P'=X(a')Z(b')$. Then
\[
   PP'=\chi(b a')X(a+a')Z(b+b')\ \text{ and }\ P'P=\chi(b'a)X(a+a')Z(b+b').
\]
As a consequence,
\[
    PP'=P'P\Longleftrightarrow \chi(b a'-b' a)=1.
\]
\end{lem}
Now we are ready to define the Pauli group; see \cite[Section 3]{me2} for the details.

%%%%%%%%%%%%%%%%%%%%%%%%%%%%%%%%%
\begin{defitheo}\label{D-Pauli}
Let $\Char(R)=c$ and let $\omega \in \C^*$ be a $\ov{c}$-primitive root of unity where 
\begin{equation}\label{e-cbar}
   \ov{c}=\left\{\!\!\begin{array}{cl}c,&\text{if $c$ is odd,}\\[.5ex] 2c,&\text{if~$c$ is even.}\end{array}\right.
\end{equation}
The set
\begin{equation}
   \cP_n:=\{\omega^\ell X(a)Z(b)\mid \ell\in\Z,\,a,b\in R^{n}\}
\end{equation}
 is a subgroup of the unitary group~$\cU(d^n)$, called the \textbf{$n$-qudit Pauli group}.
The elements of~$\cP_n$ are called \textbf{Pauli operators}.
Furthermore, the map
\begin{equation}\label{e-PsiN}
  \Psi:\cP_n\longrightarrow R^{2n},\quad \omega^\ell X(a)Z(b)\longmapsto (a,b),
\end{equation}
is a surjective group homomorphism with $\ker\Psi=\{\omega^\ell I\mid \ell\in\Z\}$.
The latter is also the center of~$\cP_n$.
\end{defitheo}
%%%%%%%%%%%%%%%%%%%%%%%%%%%%%%%%%%%

%%%%%%%%%%%%%%%%%%%%%%%%%%%%%%%%%%
\begin{defi}\label{D-Inner}
The \textbf{symplectic inner product} on $R^{2n}$ is defined as
\[
   \inners{\,\sbt\,}{\,\sbt\,}: R^{2n}\times R^{2n}\longrightarrow R,\quad \inners{(a,b)}{(a',b')}=ba'-b'a.
\]
For a subset ~$X\subseteq R^{2n}$ we define $X^\perp:=\{v\in R^{2n}\mid \inners{v}{w}=0\text{ for all }w\in X\}$.
If $X$ is a submodule of~$R^{2n}$ we call $X^\perp$ the \textbf{dual module}.
As usual, $X$ is called \textbf{self-orthogonal} (resp,\ \textbf{self-dual}) if $X\subseteq X^\perp$ (resp., $X=X^\perp$).
\end{defi}
%%%%%%%%%%%%%%%%%%%%%%%%%%%%%%%%%%

%%%%%%%%%%%%%%%%%%%%%%%%%%%%%%
\begin{prop}\label{P-ChiInner}
Let $X\subseteq R^{2n}$ be a submodule. Then
\[
  X^{\perp}=\{v\in R^{2n}\mid \chi(\inners{v}{w})=1\text{ for all }w\in X\}.
\]
\end{prop}
%%%%%%%%%%%%%%%%%%%%%%%%%%%%%%
\begin{proof}
The forward containment is obvious. The other containment follows by the fact that the kernel of a generating character does not contain any non-zero ideals; see also \cite[Prop. 3.9]{me2}.
\end{proof}

%%%%%%%%%%%%%%%%%%%%%%%%%%%%
\begin{defi}\label{D-Stab}
\begin{arabiclist}
\item A subgroup~$S$ of~$\cP_n$  is called a \textbf{stabilizer} if
      \[
         S\text{ is abelian\quad and\quad}S\cap\ker\Psi=\{I_{d^n}\}.
      \]
\item A submodule~$\cC$ of~$R^{2n}$ is called a \textbf{stabilizer code} if $\cC=\Psi(S)$ for some stabilizer~$S\leq\cP_n$.
\item A subspace~$\cQ$ of~$\C^{d^n}$ is called a \textbf{quantum stabilizer code} if there exists a stabilizer $S\leq\cP_n$ such that
      \[
         \cQ=\cQ(S):=\{v\in\C^{d^n}\mid Pv=v\text{ for all }P\in S\}.
      \]
If $\dim\cQ=1$, then~$\cQ$ is also called a \textbf{stabilizer state}.
\end{arabiclist}
\end{defi}
%%%%%%%%%%%%%%%%%%%%%%%%%%%%%%%%%%

%%%%%%%%%%%%%%%%%%%%%%%%%%
\begin{theo}\label{P-StabCodeSO}
Let $\cC\subseteq R^{2n}$ be a submodule. Then
\[
  \cC\text{ is a stabilizer code }\Longleftrightarrow \cC\subseteq \cC^\perp.
\]
Thus, the stabilizer codes are exactly the self-orthogonal submodules with respect to the symplectic inner product. In particular, stabilizer states correspond to self-dual stabilizer codes.
\end{theo}
%%%%%%%%%%%%%%%%%%%%%%%%%%%%%
\begin{proof}
The forward direction follows directly from Lemma \ref{L-MultEn} and Proposition \ref{P-ChiInner}. For the backward direction we refer the reader to \cite[Thm. 3.2]{me2}, where a stabilizer that satisfies $\Psi(S) = \cC$ is constructed. The last statement follows by the fact that $\dim_{\C}\cQ(S) = d^n/|S|$; see also \cite[Thm. 3.14]{me2}.
\end{proof}

\begin{defi}\label{D-Weight}\begin{arabiclist}
\item The \textbf{symplectic weight of a vector} $(a,b)=(a_1,\ldots,a_n,b_1,\ldots,b_n)\in R^{2n}$ is defined as
       \[
          \swt(a,b):=|\{i\mid (a_i,b_i)\neq(0,0)\}|.
       \]
The \textbf{symplectic weight of a Pauli operator} $P = \omega^lX(a)Z(b)$ is $\swt(P):=\swt(a,b)$.
\item The \textbf{minimum distance} of a  stabilizer code $\cC$ is 
\[
\dist(\cC) =  \left\{\begin{array}{ll} \!\!\min\{\swt(v) \mid v \in \cC^{\perp} - \cC\}, & \text{ if } \cC \subsetneq \cC^\perp \\ \!\!\min\{\swt(v) \mid v \in \cC - \{0\}\}, & \text{ if } \cC = \cC^\perp\end{array}.\right.
\]
\end{arabiclist}
\end{defi}

For the physical significance of the Pauli group, symplectic weight, and minimum distance we refer the reader to \cite{CRSS98,Gottesman96,KL97}. Note that by the very definition $\Psi$ is weight preserving.

\section{Symplectic Isometries}\label{Sec-Symp}
The study of symplectic isometries was initiated in \cite{me2} as a tool to understand the equivalence of quantum stabilizer codes. It was observed in \cite{me2} that not all symplectic isometries are monomial maps. A natural problem \cite[Q. 7.4]{me2} then is to establish how far from being a monomial map a symplectic isometry is. In this section we extend the study by making use of the work of Wood \cite{Wood16}. The crucial idea is to view the symplectic weight as the Hamming weight over $R^2$ and make use of classical machinery. We start with a change of coordinates that facilitates this. Namely, we use
\begin{equation}\label{e-gamma}
  \gamma: R^{2n}\longrightarrow (R^2)^n,\quad  (a_1,\ldots,a_n\mid b_1,\ldots,b_n)\longmapsto (a_1,b_1\mid a_2,b_2\mid \ldots\mid a_n,b_n).
\end{equation}
Thus for $x=(a_1,b_1\mid  \ldots\mid a_n,b_n)$ we have
\begin{equation}
\hwt(x) :=|\{i\mid (a_i,b_i) \neq (0,0)\}=\swt(\gamma^{-1}(x)),
\end{equation} that is, the Hamming weight on $(R^2)^n$ is the pullback of the symplectic weight on $R^{2n}$. In order to transfer the problem completely to $(R^2)^n$ we need to also pull back the symplectic inner product. Namely, we define
\begin{equation}\label{e-F}
  \inner{x}{y}:=\inners{\gamma^{-1}(x)}{\gamma^{-1}(y)} = \sum_{i=1}^n x_iJy_i\T, 
\end{equation}
for all $x,y\in (R^2)^n$, where $x_i, y_i \in R^2$ and
\begin{equation}
J = \left(\!\!\begin{array}{cr}0&-1\\1&0\end{array}\!\!\right).
\end{equation}
%%%%%%%%%%%%%%%%%%%%%%%%%%%%%%%%%%%%%%%%%%%%%%%%%%%%%%
\begin{defi}\label{D-Iso}
Let $\cC\subseteq R^{2n}$ be a stabilizer code and $f:\cC\longrightarrow R^{2n}$ be a linear map.
Then~$f$ is called a \textbf{symplectic isometry} if $\swt(a) = \swt(f(a))$ and
$\inners{a}{b} = \inners{f(a)}{f(b)}$ for all $a,b \in \cC$.
Two stabilizer codes $\cC\, ,\cC'\subseteq R^{2n}$ are called \textbf{symplectically isometric} if there exists a symplectic isometry
$f:\cC\longrightarrow R^{2n}$ such that $f(\cC)=\cC'$.
\end{defi}
%%%%%%%%%%%%%%%%%%%%%%%%%%%%%%%%%%%%%%%%%%%%%%%%%%%%%%%%%
For a linear map $f:R^{2n}\longrightarrow R^{2n}$ we define $\tilde{f}:=\gamma\circ f\circ\gamma^{-1}: (R^2)^n\longrightarrow(R^2)^n$ 
as in the following commutative diagram
\begin{equation}\label{e-Diagram}
\begin{array}{l}
   \begin{xy}
   (0,0)*+{(R^2)^n}="a"; (20,0)*+{(R^2)^n}="b";%
   (0,20)*+{R^{2n}}="c"; (20,20)*+{R^{2n}}="d";%
   {\ar "a";"b"}?*!/_-3mm/{\tilde{f}};
   {\ar "c";"d"}?*!/_3mm/{f};%
   {\ar "c";"a"}?*!/_-2mm/{\gamma};
   {\ar "d";"b"}?*!/_2mm/{\gamma};
   \end{xy}
\end{array}
\end{equation}
To resume, we obtain the following equivalences
\begin{equation}\label{e-ffhatw}
   f \text{ preserves $\swt$}\Longleftrightarrow \tilde{f}\text{ preserve $\hwt$},
\end{equation}
and
\begin{equation}\label{e-ffhatp}
   f \text{ preserves }\inners{\sbt}{\sbt}\Longleftrightarrow \tilde{f}\text{ preserves }\inner{\sbt}{\sbt}.
\end{equation}
We call $\tilde{f}$ a symplectic isometry if $f$ is. With this notation we obtain the structure of symplectic isometries of $R^{2n}$.

\begin{theo}[\mbox{\cite[Thm. 7.1]{me2}}]\label{T-IsoR2n}
Let $f:R^{2n}\longrightarrow R^{2n}$ be a linear map. 
Then~$f$ is a symplectic isometry iff the matrix representation of $\tilde{f}$ in $\cM_{2n}(R)$ with respect to the standard basis is a block matrix of the form
\begin{equation}\label{e-IsoR2n}
  \diag(A_1,\ldots,A_n)(P\otimes I_2),\ 
\end{equation}
$\text{ where }A_i\in\SL_2(R)\text{ and }P\in S_n$ is a permutation matrix.
\end{theo}
%%%%%%%%%%%%%%%%%%%%%%%%%%%%%%%%%%%%%%%%%%%%%%%%%%%%%%%%%%%%%%%%%

\begin{defi}
The map $\widetilde{f}$ as in \eqref{e-IsoR2n} is called a $\SL_2(R)$-\textbf{monomial map}. We will denote $\Mon_{\rm{SL}}((R^2)^n)$ the group of $\SL_2(R)$-monomial maps of $(R^2)^n$. The group of $\SL_2(R)$-monomial maps of $R^{2n}$ is given by
\begin{equation}
\Mon_{\rm{SL}}(R^{2n}):=\{\gamma^{-1}\widetilde{f}\gamma\mid\widetilde{f}\in \Mon_{\rm{SL}}((R^2)^n)\}.
\end{equation} 
The map $\widetilde{f}$ is called a \textbf{monomial map} if $A_i \in \GL_2(R)$ in \eqref{e-IsoR2n}. We will denote $\Mon((R^2)^n)$ and $\Mon(R^{2n})$ the groups of monomial maps of $(R^2)^n$ and $R^{2n}$ respectively. If two stabilizer codes are symplectially isometric via a $\SL_2(R)$-monomial map we call them \textbf{monomially equivalent}.
\end{defi}
%%%%%%%%%%%%%%%%%%%%%%%%%%%%%%%%%%%%%%%%%%%%%%%%
We will be using the term ``$(\SL_2(R)$-) monomial map" interchangeably and it should be clear from context whether we work over $(R^2)^n$ or $R^{2n}$. Theorem \ref{T-IsoR2n} implies that all the symplectic isometries of $R^{2n}$ are $\SL_2(R)$-monomial maps. On the the other hand, again thanks to Theorem \ref{T-IsoR2n}, we have that monomial maps preserve the symplectic weight, but not necessarily the symplectic inner product. 
%%%%%%%%%%%%%%%%%%%%%%%%%%%%%%%%%%%%%%%%%%%%%%%%

We have two particularly nice symplectic isometries. They are in fact $\SL_2(R)$-monomial maps, and they are naturally related with a normal form of stabilizer codes; see \cite[Thm. 4.8]{me2}.
\begin{exa}\label{Ex-Wyel}
\begin{arabiclist}
\item For every permutation $\sigma\in S_n$ define the map $\tau_\sigma: R^{2n}\longrightarrow R^{2n}$ given by
      \[
        (a_1,\ldots,a_n,b_1,\ldots,b_n)\longmapsto (a_{\sigma(1)},\ldots,a_{\sigma(n)},b_{\sigma(1)},\ldots,b_{\sigma(n)}).
      \]
It is clear that $\tau_\sigma$ is a symplectic isometry. Then, $\widetilde{\tau_\sigma}$ has matrix representation $P_\sigma\otimes I_2$. 
\item For every $i\in\{1,\ldots, n\}$ we define the map $\tau_i: R^{2n}\longrightarrow R^{2n}$ given by
      \[
        (a_1,\ldots,a_n,b_1,\ldots,b_n)\longmapsto(a_1,\ldots,a_{i-1},b_i,a_{i+1},\ldots,a_n,b_1,\ldots,b_{i-1},-a_i,b_{i+1},\ldots,b_n).
      \]
Then $\tau_i$ clearly preserves the symplectic weight. It also preserves the symplectic inner product, since for $(a,b),\,(a',b')\in R^{2n}$ we have
\[
  \inners{\tau_i(a,b)}{\tau_i(a',b')}=\sum_{j\neq i} b_ja'_j-a_ib'_i-\sum_{j\neq i} a_jb'_j+b_ia'_i
   =\sum_{j=1}^n b_ja'_j-\sum_{j=1}^n a_jb'_j=\inners{(a,b)}{(a',b')}.
\]
Moreover, the matrix representation of $\widetilde{\tau_i}$ is $\diag(I,\cdots,I,J,I,\cdots,I)\in \SL_{2n}(R)$, with $J$ at the $i$-th diagonal position.
\end{arabiclist}
\end{exa}

Theorem \ref{T-IsoR2n} heavily relies on the fact that the isometry was defined on the entire space $R^{2n}$. As we will see, the result is no longer true if we start with a stabilizer code $\cC\subseteq R^{2n}$. In particular this means that the structure of symplectic isometries between stabilizer codes is yet to be discovered. We start by defining two \textbf{isometry groups} associated to a stabilizer code $\cC\subseteq R^{2n}$:
\begin{equation}
\begin{split}
   \Mon_{\text{SL}}(\cC)&:=\{f\in \Aut(\cC)\mid f\text{ is the restriction of an $\SL_2(R)$-monomial map} \},\\
   \Symp(\cC)&:=\{f\in \Aut(\cC)\mid f\text{ is a symplectic isometry}\}.
   \end{split}
   \end{equation}
Theorem \ref{T-IsoR2n} implies $\Mon_{\text{SL}}(\cC) \subseteq \Symp(\cC)$. In fact this containment is strict, as the following example shows. See also Example \ref{E-Ex2}.
%%%%%%%%%%%%%%%%%%%%%%%%%%%%%%%%%%%%%%%%%%%%%%%%%%%%%
\begin{exa}[\mbox{\cite[Ex. 7.3]{me2}}]\label{Ex-NonEx1}
Consider the stabilizer code $\cC:=\gamma^{-1}(C) \subseteq \F_2^8$, where $C\subseteq (\F_2^2)^4$ is the $\F_2$-linear code generated by either of matrices
\[
N_1 = \left(\!\!
	\begin{array}{cc|cc|cc|cc}
    1 & 0 & 0 & 1 & 1 & 0 & 1 & 0 \\
    0 & 1 & 1 & 0 & 0 & 0 & 1 & 0 \\
    0 & 1 & 0 & 0 & 0 & 1 & 0 & 0 \\
    0 & 1 & 0 & 1 & 0 & 0 & 0 & 1
	\end{array}\!\!\right), \quad
N_2 = \left(\!\!\begin{array}{cc|cc|cc|cc}
    1 & 1 & 1 & 0 & 1 & 1 & 0 & 1 \\
    0 & 1 & 0 & 1 & 0 & 0 & 0 & 1 \\
    0 & 0 & 1 & 1 & 0 & 0 & 1 & 1 \\
    0 & 0 & 0 & 1 & 0 & 1 & 0 & 1
    \end{array}\!\!\right),
\]
and the map $\tilde{f}:C\longrightarrow C$ that sends the $i$-th row of~$N_1$ to the $i$-th row of~$N_2$. One checks straightforwardly that $\tilde{f}$ is a symplectic isometry. Moreover, $\tilde{f}$ cannot be a $\SL_2(R)$-monomial map due to the fact that there are $2\times 2$ zero blocks in $N_2$ whereas no zero blocks in $N_1$.
\end{exa}
%%%%%%%%%%%%%%%%%%%%%%%%%%%%%%%%
Since $\Mon_{\text{SL}}(\cC) \subsetneq \Symp(\cC)$, it is natural to ask how different the two groups can be. This type of question was first exploited by Wood \cite{Wood16} for classical linear codes with respect to the Hamming weight. In fact Wood showed that the difference can be as big as possible. In what follows we show that a similar scenario is true for stabilizer codes. To do so we need some preparation.

Let $\cC\subseteq \F_q^{2n}$ be a stabilizer code. Assume $\dim_{\F_q}\cC = k$ and let $G$ be a generator matrix of $\cC$, that is, $G$ is a full rank $k\times 2n$ matrix and 
\begin{equation}
\cC = \{xG\mid x\in\F_q^k\} = \im G = (\F_q^k)G.
\end{equation}
We wiew $G$ as the linear map $\F_q^k\longrightarrow \F_q^{2n}, \,x\longmapsto xG$ with inputs on the left\footnote{To avoid ambiguities, for the remainder of the section all inputs will be on the left and we precompose.}. This allows us to think of $\cC$ as an embedding of $\F_q^k$ in $\F_q^{2n}$ via $G$. That is, we identify $\cC$ with the pair $(\F_q^k, G)$. In this way, if $xG\longmapsto yG$ is an automorphism of $\cC$ then so is $xG\longmapsto yBG$ for any $B\in \GL_k(\F_q)$. In fact every isomorphism of $\cC$ is of this form. This implies 
\begin{equation}\label{e-AutC}
\Aut(\cC)=\{BG\mid B\in \GL_k(\F_q)\}.
\end{equation}
Moreover, \eqref{e-AutC} yields an isomorphism of groups
\begin{equation}\label{e-Phi}
\Phi : \Aut(\cC) \longrightarrow \GL_k(\F_q), \ f\longmapsto B_f
\end{equation}
where $B_f$ is the unique invertible matrix that satisfies $f = B_fG$. This allows us to identify $\Symp(\cC)$ with $\Phi(\Symp(\cC))\leq \GL_k(\F_q)$.
%because for any $f\in \Aut(\cC)$ there exists a unique matrix $B\in \GL_k(\F_q)$ such that the following diagram commutes.
%\begin{equation}\label{e-Diagram2}
%\begin{array}{l}
%   \begin{xy}
%   (0,0)*+{\F_q^k}="a"; (-40,20)*+{\F_q^k}="b";%
%   (0,20)*+{\cC}="c"; (-20,20)*+{\cC}="d";%
%   {\ar@{-->} "a";"b"}?*!/_3mm/{\exists !B};
%   {\ar "d";"c"}?*!/_3mm/{f};%
%   {\ar "a";"c"}?*!/_-3mm/{G};
%   {\ar "b";"d"}?*!/_3mm/{G};
%   \end{xy}
%\end{array}
%\end{equation}
An automorphism of a stabilizer code trivially preserves $\inners{\sbt}{\sbt}$. With the above identification we have 
\begin{equation}\label{e-Symp1}
\Symp(\cC) = \{B\in \GL_k(\F_q) \mid \swt(xBG) = \swt(xG) \text{ for all } x \in \F_q^k\}.
\end{equation}
Next, we address the group $\Mon_{\text{SL}}(\cC)$. Let $f \in \Mon_{\text{SL}}(\cC)$. As before, there exists a unique $B_f\in \GL_k(\F_q)$ such that $f = B_fG$. On the other hand, $f$ is the restriction of a monomial map $M$. Thus we have $B_fG = f = M_{|\cC}$.
%\begin{equation}\label{e-Diagram3}
%\begin{array}{l}
%   \begin{xy}
%   (0,0)*+{\F_q^k}="a"; (-40,20)*+{\F_q^k}="b";%
%   (0,20)*+{\cC}="c"; (-20,20)*+{\cC}="d";%
%   {\ar@{-->} "a";"b"}?*!/_3mm/{\exists !B_M};
%   {\ar "d";"c"}?*!/_3mm/{M};%
%   {\ar "a";"c"}?*!/_-3mm/{G};
%   {\ar "b";"d"}?*!/_3mm/{G};
%\end{array}
%\end{equation}
Denote\footnote{We use the same notation as Wood \cite{Wood16} where the extra ``r'' stand for ``restriction'' since we may identify $B_f$ with $M_{|\cC}$.} by
\begin{equation}
\rMon_{\text{SL}}(\cC):= \Phi(\Mon_{\rm{SL}}(\cC))\leq \GL_k(\F_q).
\end{equation}
Thus, in $\GL_k(\F_q)$ we have two subgroups that we can compare: $\rMon_{\text{SL}}(\cC)$ and $\Symp(\cC)$. Of course we have $\rMon_{\text{SL}}(\cC) \leq \Symp(\cC)$. We will show that given $H_1 \leq H_2 \leq \GL_k(\F_q)$ that satisfy some necessary conditions\footnote{Not all subgroups of $\GL_k(\F_q)$ can be isometry groups.}, there exists a stabilizer code $\cC$ such that $\rMon_{\text{SL}}(\cC)\subseteq H _1$ and $H_2 = \Symp(\cC)$, with equality $\rMon_{\text{SL}}(\cC) = H _1$ when $q=2$. We discuss first the necessary conditions following the line \cite{Wood16}. First we need the notion of closure from group theory. For more details we refer the reader to \cite{Wielandt1964} and \cite[Sec. 4]{Wood16}.
%%%%%%%%%%%%%%%%%%%%%%%%%%%%%%%%%%%%%%%%%%%
\begin{defi}
Let a group $\cG$ act on a set $X$ from the left and let $H\leq \cG$ be a subgroup. For $x\in X$, define $\text{orb}_H(x) := \{hx \mid h\in H\}$. Then the \textbf{closure} of $H$ with respect to the action of $\cG$ on $X$ is \begin{equation}
\ov{H} = \{g\in \cG\mid g\cdot\text{orb}_H(x) = \text{orb}_H(x) \text{ for all } x \in X\}.
\end{equation}
The subgroup $H$ is called \textbf{closed} if $H = \ov{H}$.
\end{defi}
%%%%%%%%%%%%%%%%%%%%%%%%%%%%%%%%%%%%%%%%%%%%
We fix the following notation for the remainder of this section.
\begin{notation}
Recall the change of coordinates $\gamma$ from \eqref{e-gamma}. Let $\cC \subseteq \F_q^{2n}$ be a stabilizer code and put $C:=\gamma(\cC) \subseteq (\F_q^2)^n$. For a generating matrix $G$ of $\cC$ we also put $N = \gamma(G)$, where the latter means that we permute the columns of $G$ accordingly. Clearly $\GL_2(\F_q)$ acts from the right on the matrix space $\cM_{k\times 2}(\F_q)$ and $\F_q^*$ acts from the left on $\F_q^k$. Denote $\cO^\#$ and $\cO$ the respective orbit spaces. The group $\GL_k(\F_q)$ acts on $\cO^\#$ from the left and on $\cO$ from the right in an obvious way. 
\end{notation}
%%%%%%%%%%%%%%%%%%%%%%%%%%%%%%%%%%%%%%%%%%%%%%%%
\begin{rem}\label{R-rem1}
Let $C \subseteq (\F_q^2)^n$ be an $\F_q$-linear code with generating matrix $H$. In this case we think of $H$ as $k\times n$ matrix whose columns are $k\times 2$ matrices. Similarly as in \eqref{e-Symp1} we may define the isometry group of $C$ as 
\begin{equation}\label{e-Symp2}
\Iso(C) := \{B\in \GL_k(\F_q) \mid \hwt(xBH) = \hwt(xH) \text{ for all } x \in \F_q^k\}.
\end{equation}
Next, let $\Mon(C):=\{f\in \Aut(C)\mid f \text{ is the restriction of a monomial map}\}$. We define $\rMon(C):=\Phi(\Mon(C))\leq \GL_k(\F_q)$. If $C$ is self-orthogonal we naturally put
\begin{equation}
\Mon_{\text{SL}}(C):=\{\widetilde{f}=\gamma\circ f \circ \gamma^{-1}\mid f\in \Mon_{\text{SL}}(\cC)\}\subseteq \Mon_{\rm{SL}}((\F_q^2)^n),
\end{equation}
where $\cC:=\gamma^{-1}(C)$. Then $\Mon_{\text{SL}}(C) \subseteq \Mon(C)$. Put $\rMon_{\text{SL}}(C):=\Phi(\Mon_{\text{SL}}(C)).$ It follows that $\rMon_{\text{SL}}(\cC) = \rMon_{\text{SL}}(C)$.
\end{rem}
%%%%%%%%%%%%%%%%%%%%%%%%%%%%%%%%%%%%
\begin{rem}\label{R-rem2}
Let $C\subseteq (\F_q^2)^n = \im H$ be a self-orthogonal $\F_q$-linear code and put $\cC:=\gamma^{-1}(C) = \im G$. Then $\swt(xG) = \hwt(xH)$ for all $x\in \F_q^k$. Comparing \eqref{e-Symp1} and \eqref{e-Symp2} we conclude that $\Iso(C) = \Symp(\cC)$. In addition, Remark \ref{R-rem1} implies $\rMon_{\text{SL}}(\cC)=\rMon_{\text{SL}}(C)\subseteq \rMon(C)$. When $q =2$ we have $\GL_2(\F_2) = \SL_2(\F_2)$ and thus $\rMon_{\text{SL}}(\cC) = \rMon(C)$.
\end{rem}
%%%%%%%%%%%%%%%%%%%%%%%%%%%%%%%%%%%%%%%%%%%
Remarks \ref{R-rem1} and \ref{R-rem2} point out the importance of the isomorphism $\Phi$ from \eqref{e-Phi}. By considering the images under $\Phi$ of all the groups floating around we obtain a unified approach that is independent of the change of coordinates $\gamma$.
\begin{exa}\label{E-Ex2}
Let $\cC\subseteq \F_2^{2\cdot 5}$ be the stabilizer given by the following generating matrix
\[
G = \left( \!\!\begin{array}{ccccc|ccccc} 0&1&1&1&1&0&0&0&0&0\\ 1&0&1&0&0&0&0&0&1&1\\ 1&0&0&0&1&0&1&1&0&0\end{array}\!\!\right).
\]
Using \eqref{e-Symp1} one computes $\Symp(\cC) = \GL_3(\F_2)$. On the other hand, only 8 of these symplectic isometries are restrictions of $\SL_2(\F_2)$-monomial maps.
\end{exa}
%%%%%%%%%%%%%%%%%%%%%%%%%%%%%%%%%%%%%%%%%%%%%%
Then, \cite[Prop. 4.7]{Wood16} applied to our specific scenario reduces to the following.
%%%%%%%%%%%%%%%%%%%%%%%%%%%%%%%%%%%%%%%%%%%%%%%%%%%
\begin{prop}
Let $C\subseteq (\F_q^2)^n$ be a $\F_q$-linear self-orthogonal code of dimension $k$. Then $\rMon(C)$ is closed with respect to the action of $\GL_k(\F_q)$ on $\cO^\#$ and $\Iso(C)$ is closed with respect to the action of $\GL_k(\F_q)$ on $\cO$.
\end{prop}

%%%%%%%%%%%%%%%%%%%%%%%%%%%%%%%%%%%%%%%%%%
\begin{theo}[\mbox{\cite[Thm. 5.1]{Wood16}}]\label{T-WoodMain}
Let $H_1,\, H_2\leq \GL_k(\F_q)$ be two subgroups such that $H_1$ is closed under the action of $\GL_k(\F_q)$ on $\cO^\#$ and $H_2$ is closed under the action of $\GL_k(\F_q)$ on $\cO$. Then there exists $n \in \N$ and a $\F_q$-linear code $C\subseteq (\F_q^2)^n$ such that 
\[
H_1 = \rMon(C) \text{ and } H_2 = \Iso(C).
\]
\end{theo}
%%%%%%%%%%%%%%%%%%%%%%%%%%%%%%
Of course there is no reason for the linear code produced in Theorem \ref{T-WoodMain} to be self-orthogonal. However, we make use of it to produce a self-orthogonal code of the same dimension without changing the isometry groups. To achieve this we make use of the concatenated code. That is, for a linear code $C$, the \textbf{concatenated code} is defined as\begin{equation}
C\mid C :=\{(x\mid x) \mid x\in C\}\subseteq (\F_q^2)^{2n}.
\end{equation}
Clearly, $C\mid C$ has the same dimension as $C$, but it is twice as long. In this sense $C$ has a rate twice as large as the rate of $C\mid C$. So of course, achieving self-orthogonality will come with a high cost. 
%%%%%%%%%%%%%%%%%%%%%%%%%%%%%%%%%%%%%%%%%%%%%%%%%%%%%
\begin{lem}
Let $C = \im N\subseteq (\F_q^2)^{n}$ be a $\F_q$-linear code. Then $\rMon(C\mid C) = \rMon(C)$ and $\Iso(C\mid C) = \Iso(C).$
\end{lem}
\begin{proof}
The first statement is a corollary of \cite[Prop. 3.7]{Wood16} along with the observation that 
%Recall that $\rMon(C) = \Phi(\Mon(C))$. Therefore, the first task is to establish $\Mon(C\mid C)$. We compute
%\begin{align*}
%\Mon(C\mid C) &= \{A\in \GL_{4n}(\F_q) \mid (C\mid C)A  = (C\mid C) \text{ and } A \text{ monomial}\} \\
%& = \left\{\widetilde{M}:=\left(\!\!\!\begin{array}{cc} M&0\\0&M\end{array}\!\!\!\right)\,\middle|\, M\in \Mon(C)\right\}.
%\end{align*}
$C\mid C = \im \widehat{N}$ where $\widehat{N}:= N\mid N$ is the corresponding concatenated matrix.
%Assume $C$ has dimension $k$ over $\F_q$. Then for all $B\in \GL_k(\F_q)$ we have 
%\begin{equation}
%BN = BM_{|C} \iff B\widehat{N} = (BN\mid BN) =  (BM_{|C}\mid BM_{|C}) = B\widetilde{M}_{|(C\mid C)},
%\end{equation}
%where the last equality follows by following obvious equation
%\begin{equation}
%M_{|C}\mid M_{|C} = \widetilde{M}_{|(C\mid C)}.
%\end{equation} 
%By the very definition of $\Phi$, and following \eqref{e-Diagram3} (recall that we precompose the maps and inputs are on the left) we obtain $\rMon(C\mid C) = \rMon(C)$. 
Next, by the very definition of the Hamming weight, for all $B\in \GL_k(\F_q)$ we have
\[
\hwt(xN\mid xN) = \hwt(xBN\mid xBN) \iff \hwt(xN) = \hwt(xBN).
\]
The second statement then follows.
\end{proof}
%%%%%%%%%%%%%%%%%%%%%%%%%%%%%%%%%%%%%%%%%%%%%%%
%%%%%%%%%%%%%%%%%%%%%%%%%%%%%%%%%%%%%%%%%%%%%%%
\begin{lem}
Fix $q= 2^l$. Let $C \subseteq (\F_q^2)^{n}$ be a $\F_q$-linear code. Then $C\mid C \subseteq (\F_q^2)^{2n}$ is a self-orthogonal $\F_q$-linear code. 
\end{lem}
\begin{proof}
Let $x = (x_1,\ldots,x_n), y=(y_1,\ldots, y_n) \in (\F_q^2)^n$. Then 
\[
\inner{(x\mid x)}{(y\mid y)} = \sum_{i=1}^n x_iJy_i\T + \sum_{i=1}^n x_iJy_i\T = 0,
\]
since $\Char(\F_q) = 2$. Thus $C\mid C$ is self-orthogonal. 
\end{proof}
%%%%%%%%%%%%%%%%%%%%%%%%%%%%%%%%%%%%%
\begin{cor}\label{C-cor}
Let $C \subseteq (\F_q^2)^n$ be a $\F_q$-linear code where $q = p^l$ for some prime $p$. Then the $p$-th concatenated code $\widetilde{C}:=C\mid\cdots\mid C\subseteq (\F_q^2)^{pn}$ is self-orthogonal code such that $\rMon(C) = \rMon(\widetilde{C})$ and $\Iso(C) = \Iso(\widetilde{C})$.
\end{cor}
%%%%%%%%%%%%%%%%%%%%%%%%%%%%%%%%

We are now ready to prove the main theorem.

\begin{theo}\label{T-main}
Let $H_1,\, H_2\leq \GL_k(\F_q)$ be two subgroups such that $H_1$ is closed under the action of $\GL_k(\F_q)$ on $\cO^\#$ and $H_2$ is closed under the action of $\GL_k(\F_q)$ on $\cO$. Then there exists $n \in \N$ and a stabilizer code $\cC\subseteq \F_q^{2n}$ such that 
\begin{equation}\label{e-main}
\rMon_{\text{SL}}(\cC)\subseteq H_1 \text{ and } H_2 = \Symp(\cC),
\end{equation}
with equality $H_1=\rMon_{\text{SL}}(\cC)$ if $q = 2$.
\end{theo}
\begin{proof}
Applying Corollary \ref{C-cor} to Theorem \ref{T-WoodMain} we can produce a self-orthogonal code $C\subseteq (\F_q^2)^n$, for some $n$, such that 
\[
H_1 = \rMon(C) \text{ and } H_2 = \Iso(C).
\]
Now $\cC:=\gamma^{-1}(C)\subseteq \F_q^{2n}$ is a stabilizer code that satisfies \eqref{e-main}, thanks to Remark \ref{R-rem2}. The equality for the case $q=2$ was also discussed in Remarks \ref{R-rem1} and \ref{R-rem2}.
\end{proof}

We now address the general case of stabilizer codes over a local commutative Frobenius $R$. In this case we obtain a weaker version of Theorem \ref{T-main}. Let $\mm$ be the maximal ideal and $\alpha$ a generator of the socle. Recall from Remark \ref{R-LFR} that $\soc(R) = \alpha R \cong R/\mm=\F_q$.  

\begin{rem}\label{R-rem3}
Let $X\subseteq R^{2n}$ be a subset. We denote $\alpha X:=\{\alpha x \mid x \in X\}$ and $\ov{X}:=\{\ov{x}\mid x \in X\}\subseteq \F_q^{2n}$. Note that $\alpha X$ is trivially self-orthogonal. Thus, $\alpha X$ is a stabilizer code for any submodule $X\subseteq R^{2n}$. Recall the map $\rho$ from Remark \ref{R-LFR}. It induces a map, called again $\rho$, $\alpha R^{2n} \longrightarrow \F_q^{2n}$. Thus $\alpha X \cong \ov{X}$ for any submodule $X\subseteq R^{2n}$, where the isomorphism is $R$-linear and $\F_q$-linear.
\end{rem}
%%%%%%%%%%%%%%%%%%%%%%%%%%%%%%%%%%%%%%%%%%%%%%
%In this case we obtain a weaker version of Theorem \ref{T-main}.
\begin{theo}\label{T-weaker}
Let $H\leq \GL_k(\F_q)$ be a closed subgroup under the action of $\GL_k(\F_q)$ on $\cO$. Then there exists $n \in \N$ and a stabilizer code $\cC\subseteq R^{2n}$ such that $H = \Symp(\cC)$.
\end{theo}

\begin{proof}
Let $\widetilde{\cC} \subseteq \F_q^{2n}$ be the stabilizer code produced by Theorem \ref{T-main} that satisfies $H = \Symp(\widetilde{\cC})$. Write $\widetilde{\cC} = \im G$, and let $\ov{g_i}$ be the $i$th row of $G$. Then $\cC:=\rho^{-1}(\widetilde{\cC}) \subseteq (\alpha R)^{2n} \subseteq R^{2n}$ is a stabilizer code over $R$ thanks to Remark \ref{R-rem3}. Let $G'$ be the matrix whose $i$th row is $\alpha g_i$. Thanks to \eqref{e-FqVS} we have 
\begin{equation}
\cC = \{xG'\mid \ov{x}\in \F_q^k\}=\im_{\!\F_q}G'.
\end{equation}
Furthermore, $\F_q$-linear automorphisms of $\cC$ and $R$-linear automorphisms coincide thanks to Remark \ref{R-LFR}. This implies $\Symp(\cC) = \Symp(\widetilde{\cC}) = H$. 
\end{proof}

%\begin{rem}\label{R-rem5}
%Let $\cC$ be as in the proof Theorem \ref{T-weaker}. Then $\cC$ can be viewed simultaneously as a stabilizer code over $R$ and over $\F_q$. Let $\Mon_{\text{SL},\,R}(\cC)$ and $\Mon_{\text{SL},\,\F_q}(\cC)$ denote the monomial groups of $\cC$ over $R$ and over $\F_q$ respectively. It is clear that every $\SL_2(\F_q)$-monomial map is a $\SL_2(R)$-monomial map. This yields $\Mon_{\text{SL},\,\F_q}(\cC)\subseteq \Mon_{\text{SL},\,R}(\cC)$, and thus $\rMon_{\F_q}(\cC)\subseteq \rMon_R(\cC)$.
%\end{rem}

So far we have been comparing symplectic isometries of stabilizer codes with the symplectic isometries of the entire ambient space. But for a stabilizer code $\cC$ we have $\cC\subseteq \cC^\perp$. How do symplectic isometries of $\cC$ relate to symplectic isometries of $\cC^\perp$? We end this section with an example that addresses this. See also Questions 7.5 and 7.6 in \cite{me2}.
\begin{exa}\label{Ex-Extension2}
Consider the stabilizer code $\cC :=\gamma^{-1}(C)$ where $C \subseteq (\F_2^2)^4$ is the self-orthogonal code generated by the matrix
\[
G = \left(\!\!
	\begin{array}{cc|cc|cc|cc}
    1 & 0 & 0 & 0 & 0 & 0 & 1 & 1 \\
    0 & 0 & 1 & 0 & 0 & 0 & 0 & 0 \\
    0 & 1 & 0 & 0 & 1 & 0 & 1 & 0 
	\end{array}\!\!\right).
\]
It is easy to see that $\cC^\perp$ is generated by the matrix
\[
H = \left(\!\!
	\begin{array}{cc|cc|cc|cc}
    1 & 0 & 0 & 0 & 0 & 0 & 1 & 1 \\
    0 & 0 & 1 & 0 & 0 & 0 & 0 & 0 \\
    0 & 1 & 0 & 0 & 1 & 0 & 1 & 0 \\
    0 & 1 & 0 & 0 & 0 & 0 & 1 & 0 \\
    0 & 1 & 0 & 0 & 0 & 1 & 0 & 1
	\end{array}\!\!\right).
\]
Let $g_i$ be the $i$-th row of $G$ and $f:C \longrightarrow (\F_2^2)^4$ be the symplectic isometry given by 
\begin{equation}
\begin{split}
g_1 \longmapsto \begin{array}{cc|cc|cc|cc}(1&0&0&0&0&0&1&0)\end{array}\\
g_2 \longmapsto \begin{array}{cc|cc|cc|cc}(0&0&0&1&0&0&0&0)\end{array}\\
g_3 \longmapsto \begin{array}{cc|cc|cc|cc}(0&1&0&0&1&0&0&1)\end{array}\\
\end{split}
\end{equation}
Clearly, there are exactly three self-dual codes $C_i$ such that $C\subsetneq C_i \subsetneq C^\perp$. Namely, if $h_i$ is the $i$-th row of $H$, they are $C\oplus \group{h_4}, C\oplus \group{h_5}$, and $C\oplus\group{h_4+h_5}$. We claim that $f$ cannot be extended to a symplectic isometry of $C^\perp$. To that end, assume $f$ extends to a linear map $C^\perp \longrightarrow \F_2^8$ that preserves orthogonality with respect to $\inner{\sbt}{\sbt}$, called again $f$. Since $C_i$'s are self-dual so are $f(C_i)$'s. Put $\widetilde{C}:=f(C)$. Then $\widetilde{C}^\perp$ has generating matrix 
\[
\widetilde{H} = \left(\!\!
	\begin{array}{cc|cc|cc|cc}
    1 & 0 & 0 & 0 & 0 & 0 & 1 & 0 \\
    0 & 0 & 0 & 1 & 0 & 0 & 0 & 0 \\
    0 & 1 & 0 & 0 & 1 & 0 & 0 & 1 \\
    0 & 1 & 0 & 0 & 0 & 0 & 0 & 1 \\
    1 & 0 & 0 & 0 & 0 & 1 & 0 & 0
	\end{array}\!\!\right).
\]
Similarly, there are three self-dual codes $\widetilde{C}_j$ such that $\widetilde{C}\subsetneq \widetilde{C}_j \subsetneq \widetilde{C}^\perp.$ Namely, if $\widetilde{h}_i$ is the $i$-th row of $\widetilde{H}$, they are $\widetilde{C}\oplus \group{\widetilde{h}_4}, \widetilde{C}\oplus \group{\widetilde{h}_5}$, and $\widetilde{C}\oplus\group{\widetilde{h}_4+\widetilde{h}_5}$. Thus $f(C_i) = \widetilde{C}_j$ for some $j$, and $f(C_i - C) = \widetilde{C}_j - \widetilde{C}$. \noindent By comparing the weight-distributions of $C_i - C$ and $\widetilde{C}_j - \widetilde{C}$ for all $i,j$, we must have $f(C_1) = \widetilde{C}_1$ in order to preserve the Hamming weight. By the same argument $f$ cannot be extended any further. %From Table 1 and 2 we can also compare the weight distributions of $C^\perp - C_1$ and $\widetilde{C}^\perp -\widetilde{C}_1$ and deduce that $f:C^\perp \longrightarrow \widetilde{C}^\perp$ does not preserve the Hamming weight. 

\begin{table}[h]
\centering
 \begin{tabular}{c|| c ||c ||c} 
  \diaghead{\theadfont gfdgdfgfdssd}%
{$C$}{$C_i$}& $01|00|00|10$ \circled{2} & $01|00|01|01$ \circled{3}& $00|00|01|11$ \circled{2}\\ [0.5ex] 
 \hline\hline
 $10|00|00|11$ & $11|00|00|01$ \circled{2}& $11|00|01|10$ \circled{3}& $10|00|01|00$ \circled{2}\\ 
 \hline
 $00|10|00|00$ & $01|10|00|10$ \circled{3}& $01|10|01|01$ \circled{4}& $00|10|01|11$ \circled{3}\\
 \hline
 $01|00|10|10$ & $00|00|10|00$ \circled{1}& $00|00|11|11$ \circled{2}& $01|00|11|01$ \circled{3}\\
 \hline
 $10|10|00|11$ & $11|10|00|01$ \circled{3}& $11|10|01|10$ \circled{4}& $10|10|01|00$ \circled{3}\\
 \hline
 $01|10|10|10$ & $00|10|10|00$ \circled{2}& $00|10|11|11$ \circled{3}& $01|10|11|01$ \circled{4}\\ \hline
 $11|00|10|01$ & $10|00|10|11$ \circled{3}& $10|00|11|00$ \circled{2}& $11|00|11|10$ \circled{3}\\
 \hline
 $11|10|10|01$ & $10|10|10|11$ \circled{4}& $10|10|11|00$ \circled{3}& $11|10|11|10$ \circled{4}
\end{tabular}
\vspace{-.1 in}
\caption{Weight distributions of $C_i-C$}
\label{Table:1}
\end{table}

\begin{table}[h]
\centering
 \begin{tabular}{c|| c ||c ||c} 
  \diaghead{\theadfont gfdgdfgfdssd}%
{$\widetilde{C}$}{$\widetilde{C}_j$}& $01|00|00|01$ \circled{2} & $10|00|01|00$ \circled{2}& $11|00|01|00$ \circled{2}\\ [0.5ex] 
 \hline\hline
 $10|00|00|10$ & $11|00|00|11$ \circled{2}& $00|00|01|10$ \circled{2}& $01|00|01|10$ \circled{3}\\ 
 \hline
 $00|01|00|00$ & $01|01|00|01$ \circled{3}& $10|01|01|00$ \circled{3}& $11|01|01|00$ \circled{3}\\
 \hline
 $01|00|10|01$ & $00|00|10|00$ \circled{1}& $11|00|11|01$ \circled{3}& $10|00|11|01$ \circled{3}\\
 \hline
 $10|01|00|10$ & $11|01|00|11$ \circled{3}& $00|01|01|10$ \circled{3}& $01|01|01|10$ \circled{4}\\
 \hline
 $01|01|10|01$ & $00|01|10|00$ \circled{2}& $11|01|11|01$ \circled{4}& $10|01|11|01$ \circled{4}\\ \hline
 $11|00|10|11$ & $10|00|10|10$ \circled{3}& $10|00|11|11$ \circled{3}& $00|00|11|11$ \circled{2}\\
 \hline
 $11|01|10|11$ & $10|01|10|10$ \circled{4}& $01|01|11|11$ \circled{4}& $00|01|11|11$ \circled{3}
\end{tabular}
\vspace{-.1 in}
\caption{Weight distributions of $\widetilde{C}_j-\widetilde{C}$}
\label{Table:2}
\end{table}
\end{exa}

\section{Applications to LU-LC Conjecture}\label{Sec-LULC}
The Pauli group is by definition a subgroup of the unitary group $\cU(d^n)$. For a unitary matrix $U\in \cU(d^n)$ we have $U^\dagger = U^{-1}$ where the dagger represents the conjugate transpose. Thus the normalizer of the Pauli group is given by 
\begin{equation}
\cN(\cP_n) : = \{U\in \cU(d^q) \mid U\cP_nU^\dagger = \cP_n\}.
\end{equation}
\begin{defi}
The $n$-\textbf{qudit Clifford group} is $\cC_n:=\cN(\cP_n)/\{e^{i\theta}I\mid \theta\in \R\}$.
\end{defi}
Note that the Clifford group is simply the normalizer of the Pauli group where we disregard the phases. The latter is of course justified by \textbf{phase principle} which in quantum computation has no physical consequence. Throughout this section we will pay special attention to the subgroup $\cC_1^{\otimes n}\leq\cC_n$. We call $U\in \cC_n$ a \textbf{Clifford operator} whereas $U\in \cC_1^{\otimes n}$ a \textbf{local Clifford} (LC) operator. Recall the surjective group homomorphism $\Psi$ from \eqref{e-PsiN}, with kernel $\ker \Psi = \{\omega^\ell I\mid\ell\in \Z\}$. We will denote $\cP_n^*:=\cP_n/\ker \Psi$. Thus we have an induced isomorphism 
\begin{equation}\label{e-Psi*}
\Psi^*:\cP_n^* \longrightarrow R^{2n}.
\end{equation}
Then $\Psi$ and $\Psi^*$ agree when restricted to stabilizers. The normalizer $\cN(\cP_n)$ acts on $\cP_n$ via conjugation. This induces a well-defined action of $\cC_n$ on $\cP_n^*$. Stated differently, for all $U \in \cC_n$ we obtain a group homomorphism
\begin{equation}\label{e-Psi*U}
\phi_U:\cP_n^*\longrightarrow \cP_n^*, \,P\longmapsto UPU^\dagger,
\end{equation}
which in turn is an automorphism of $\cP_n^*$. 
%%%%%%%%%%%%%%%%%%%%%%%%%%%%%%%%%%%%%%%%
\begin{rem}\label{rem7}
Similarly as above, using the action of $\cN(\cP_n)$ on $\cP_n$ we also obtain a group homomorphism
\begin{equation}
  \Phi: \cN(\cP_n)\longmapsto \Aut(\cP_n), \,\,\,\\
   U\longmapsto \left\{\!\!\begin{array}{rcl}\Phi_U:\cP_n&\longrightarrow& \cP_n\\ P&\longmapsto&UPU^\dagger\end{array}\right..
\end{equation}
Note that $U \in \ker \Phi$ iff $U$ commutes with every Pauli operator. Since the Pauli operators span\footnote{For instance, see \cite[Rem. 3.6]{me2} and the references therein.} the matrix space $\cM_{d^n}(\C)$, we may conclude that 
\begin{align*}
U \in \ker \Phi & \iff UM = MU \text{ for all }M\in \cM_{d^n}(\C)\\
&\iff U \in \{e^{i\theta}I\mid \theta \in \R\}.
\end{align*}
Hence $\cC_n = \cN(\cP_n)/\ker\Phi$ can be thought of as a subgroup of $\Aut(\cP_n)$. Namely,
\begin{equation}\label{e-Pn*}
\cC_n \cong \{\Phi_U\mid U\in \cN(\cP_n)\} \leq \Aut(\cP_n).
\end{equation}
\end{rem}

Although Remark \ref{rem7} gives a natural connection of the Clifford group with automorphisms of the Pauli group, we focus only on \eqref{e-Psi*} and \eqref{e-Psi*U}. Thanks to \eqref{e-Psi*} we clearly have $\Aut(\cP_n^*) \cong \Aut(R^{2n})$. Moreover, the map $\Psi_U:=\Psi^{*-1}\circ \phi_U\circ\Psi^*$ is an automorphism of the additive group $(R^{2n}, +)$ for any $U \in \cC_n$. Since $\Psi^*$ and $\phi_U$ are only group isomorphisms, it is impossible to say anything about $R$-linearity of $\Psi_U$. For this reason we restrict ourselves to the Frobenius ring $R:=\Z/d\Z$. With this restriction, $\Psi_U$ is $R$-linear and it is given by right matrix multiplication. Namely, for a matrix $M \in \GL_{2n}(R)$ denote $L_M:x\longmapsto xM$ its induced linear map. Then for every $U\in \cC_n$ there exists $M(U)\in \GL_{2n}(R)$ such that the following diagram
\begin{equation}\label{e-Diagram1}
\begin{array}{l}
   \begin{xy}
   (0,0)*+{R^{2n}}="a"; (20,0)*+{R^{2n}}="b";%
   (0,20)*+{\cP_n^*}="c"; (20,20)*+{\cP_n^*}="d";%
   {\ar "a";"b"}?*!/_-3mm/{L_{M(U)}};
   {\ar "c";"d"}?*!/_3mm/{\phi_{U}};%
   {\ar "c";"a"}?*!/_-3mm/{\Psi^*};
   {\ar "d";"b"}?*!/_3mm/{\Psi^*};
   \end{xy}
\end{array}
\end{equation}
commutes.
\begin{rem}\label{R-MON}
Consider \eqref{e-Diagram1} for $n=1$ and recall that we have fixed $R:=\Z/d\Z$. It is straightforward to show that for every $U\in \cC_1$ we have $M(U)\in \SL_2(R)$. The converse is also true, that is, 
\begin{equation}\label{e-dmod}
\text{for every $M\in \SL_2(R)$, there exists $U(M) \in \cC_1$ such that \eqref{e-Diagram1} commutes}.
\end{equation}
In this paper we will need only the existence, thus, for the details of the existence we refer the reader to \cite{HostensEtAl05,Appleby05}. It is worth mentioning that in these references the arithmetic is modulo $\ov{d}$ where $\ov{d}$ is as in \eqref{e-cbar}. Then one shows that the same holds true modulo $d$; see \cite[Lemma A.1]{ABBGL11}, for instance. Hence, \eqref{e-dmod} holds regardless of whether $d$ is odd or even. Now let $U = U_1\otimes\cdots\otimes U_n\in \cC_1^{\otimes n}$. Then 
\begin{equation}
M(U) = \diag(M(U_i))_i
\end{equation}
is a $2n\times 2n$ block diagonal matrix, where $M(U_i)\in \SL_2(R)$. In other words, $M(U)$ is a $\SL_2(R)$-monomial map for every $U\in \cC_1^{\otimes n}$.
\end{rem}
%%%%%%%%%%%%%%%%%%%%%%%%%%%%%%%%%%%%%%%%%%%%%%%%
\begin{rem}\label{R-USU}
Let $S\leq \cP_n$ be a stabilizer. By definition $S \cap \ker \Psi = \{I\}$ and thus $\Psi(S) = \Psi^*(S)$ gives rise to a stabilizer code $\cC\subseteq R^{2n}$. It is easy to see that for any $U\in \cC_n$ the group
\begin{equation}\label{e-UPU}
\phi_U(S) = USU^\dagger := \{UPU^\dagger\mid P\in S\}
\end{equation}
is again a stabilizer. Thus $\Psi(USU^\dagger)$ also defines a stabilizer code $\cC_U\subseteq R^{2n}$. Moreover, we obtain a quantum stabilizer code $\cQ(USU^\dagger)$. The reader will verify that $\cQ(USU^\dagger) = U\cQ(S) :=\{Uv\mid v\in \cQ(S)\}$.
\end{rem}
%%%%%%%%%%%%%%%%%%%%%%%%%%%%%%%%%%%%%%%%%%%%%%%%%
\begin{theo}
Let $U\in\cC_1^{\otimes n}$. Then $\cC$ and $\cC_U$ as in Remark \ref{R-USU}. are symplectically isometric.
\end{theo}
\begin{proof}
Write $U = U_1\otimes \cdots \otimes U_n$ with $U_i \in \cC_1$. Consider the map $\Psi_U:=\Psi^{*-1}\phi_U\Psi^*$. By Remark \ref{R-USU} we have $\Psi_U(\cC) = \cC_U$. Thus, $\Psi_U$ trivially preserves the symplectic inner product on $\cC$. To complete the proof we need to show that $\Psi_U$ also preserves the symplectic weight. Since $\Psi$ is a weight preserving map, it is enough to show that $\phi_U$ is weight preserving for any $U = U_1\otimes\cdots\otimes U_n\in \cC_1^{\otimes n}$. Indeed, let $P=P_1\otimes \cdots \otimes P_n\in S$. From the very definition of the symplectic weight of a Pauli operator we have $\swt(P) = |\{i\mid P_i\neq I\}|$. Moreover, we have 
\begin{equation}
\phi_U(P) = UPU^\dagger = U_1P_1U_1^\dagger\otimes \cdots \otimes U_nP_nU_n^\dagger,
\end{equation}
which in turn implies $\swt(P) = \swt(\phi_U(P))$.
\end{proof}

\begin{notation} \label{n-notation2}
A permutation $\sigma \in S_n$ acts on $R^n$ by permuting the coordinates. For $P = \omega^lX(a)Z(b)$ we will denote $\sigma(P):= \omega^lX(\sigma(a))Z(\sigma(b))$ and for $X\subseteq \cP_n$ we will denote $\sigma(X):=\{\sigma(x)\mid x\in X\}$. It is easy to see that $S \leq \cP_n$ is a stabilizer iff $\sigma(S) \leq \cP_n$ is a stabilizer.
\end{notation}

\begin{defi}\label{D-LCP}
\begin{arabiclist}
\item Two quantum stabilizer codes $\cQ = \cQ(S)$ and $\cQ=\cQ(S')$ are called \textbf{permutation} \textbf{equivalent} if there exists a permutation $\sigma \in S_n$ such that $S' = \sigma(S)$. 
\item Two quantum stabilizer codes $\cQ = \cQ(S)$ and $\cQ=\cQ(S')$ are called \textbf{Clifford permutation equivalent} (CP) (resp., \textbf{locally Clifford permutation equivalent} (LCP)) if  there exists a permutation $\sigma \in S_n$ and $U\in \cC_n$ (resp., $U\in \cC_1^{\otimes n}$) such that $S' = U\sigma(S)U^{\dagger}$. 
\item Two quantum stabilizer codes $\cQ$ and $\cQ'$ are called \textbf{unitary equivalent} (resp., \textbf{locally unitary equivalent} (LU)) if there exists $U\in \cU(d^n)$ (resp., $U\in \cU(d)^{\otimes n}$) such that $\cQ'=U\cQ$.
\end{arabiclist}
\end{defi}
%%%%%%%%%%%%%%%%%%%%%%%%%%%%%%%%%%%%%
If we take $\sigma$ to be the identity permutation in Definition \ref{D-LCP}(2) then we are dealing with \textbf{locally Clifford} (LC) equivalent quantum stabilizer codes. It is obvious that two LC equivalent quantum stabilizer codes are also LU equivalent. Is the converse true? This is know in the literature as the LU-LC conjecture \cite{KW05}. The conjecture was reduced to various subclasses of stabilizer codes \cite{GVN08,NNMVDD03,Zeng06,Nest05}, to finally be proven incorrect in \cite{LULC-false}. One of these subclasses is that of stabilizer states, to which correspond self-dual stabilizer codes. The counterexample provided in \cite{LULC-false} is randomly generated. Thus the structure of such counterexamples is yet to be discovered. In \cite{SR10} the authors show that there exist infinitely many stabilizer states that disprove the LU-LC conjecture. A sufficient condition for spotting LU stabilizer states that are not LC is of interest. The following result characterizes LCP stabilizer codes (and thus LC stabilizer states) using the language of Section \ref{Sec-Symp}.
%%%%%%%%%%%%%%%%%%%%%%%%%%%%%%%%%%%%
\begin{theo}\label{T-LCP1}
Let $\cC=\Psi(S)$ and $\cC' = \Psi(S')$ be two stabilizer codes. Then $\cC$ and $\cC'$ are monomially equivalent iff the quantum stabilizer codes $\cQ(S)$ and $\cQ(S')$ are LCP equivalent. 
\end{theo}

\begin{proof}
We show the forward direction, with the other one being similar. Let $M = \diag(M_1,\ldots, M_n)(P_{\sigma}\otimes I_2)$ be a $\SL_2(R)$-monomial map as in \eqref{e-IsoR2n} that maps $\cC$ to $\cC'$. Let $U_i:=U(M_i)\in \cC_1$ be as in Remark \ref{R-MON} and consider $U:=U_1\otimes\cdots\otimes U_n\in \cC_1^{\otimes n}$. Recall the change of coordinates $\gamma$ from \eqref{e-gamma}. For $(a,b)\in \cC$ we have 
\begin{equation}
\gamma(a,b)=:x=(x_1,\ldots,x_n)\in \gamma(\cC)=:C\subseteq (R^2)^n,
\end{equation}
where $x_i=(a_i,b_i)\in R^2$. Put $P_i = \Psi^{*-1}(x_i)$. Then $P=P_1\otimes\cdots \otimes P_n\in S$, and every element of $S$ can be written in such way. With this notation we have 
\begin{align*}
U\sigma(P)U^\dagger & = U_1P_{\sigma(1)}U_1^\dagger\otimes \cdots \otimes U_nP_{\sigma(n)}U_n^\dagger \\
& = \phi_{U_1}(P_{\sigma(1)})\otimes\cdots\otimes \phi_{U_n}(P_{\sigma(n)})\\
& = \phi_{U_1}(\Psi^{*-1}(x_{\sigma(1)}))\otimes\cdots\otimes \phi_{U_n}(\Psi^{*-1}(x_{\sigma(n)}))\\
& = \Psi^{*-1}(x_{\sigma(1)}M_1)\otimes\cdots\otimes\Psi^{*-1}(x_{\sigma(n)}M_n)\\
&\in S',
\end{align*}
because $(x_{\sigma(1)}M_1,\ldots,x_{\sigma(n)}M_n) \in \gamma(\cC')$. Thus $U\sigma(S)U^\dagger \subseteq S'$. Since $|S'| = |\cC'|=|\cC|=|S|=|U\sigma(S)U^\dagger|$, equality follows.
\end{proof}

We end this section with two examples that relate all the equivalence notions discussed. Throughout we will use $R= \F_2$ and $X:=X(1),\,Z:=Z(1)$.

\begin{exa}\label{Ex-Ex11}
Let $\cC\subseteq \F_2^{2\cdot 3}$ be the stabilizer code given by the following generating matrix

$$
G = \left(\!\!\begin{array}{ccc|ccc}
1 & 0 & 1 & 0 & 1 & 0  \\
0 & 1 & 1 & 1 & 0 & 0 \\
0 & 0 & 0 & 1 & 1 & 1
\end{array}\!\!\right),
$$\\
and consider the $\SL_2(\F_2)$-monomial map given by $M = \diag(M_1,M_2,M_3)(P_{\sigma}\otimes I_2)$ where we take the permutation to be the cycle $\sigma = (123)$, and
$$ M_1 = \left(\!\!\begin{array}{cc}
1 & 0 \\
1 & 1 
\end{array}\!\!\right),
M_2=\left(\!\!\begin{array}{cc}
0 & 1 \\
1 & 0 
\end{array}\!\!\right),
M_3=\left(\!\!\begin{array}{cc}
1 & 1 \\
0 & 1 
\end{array}\!\!\right)
.$$ \\
Then, $\cC':=\{xM\mid x\in\cC\}$ is the stabilizer code given the following generating matrix
$$
G' = \left(\!\!\begin{array}{ccc|ccc}
1 & 0 & 1 & 1 & 1 & 1  \\
1 & 0 & 0 & 0 & 1 & 1 \\
1 & 1 & 0 & 1 & 0 & 1
\end{array}\!\!\right).
$$
Then the corresponding stabilizers are $S = \langle XZX,ZXX,ZZZ \rangle$ and $S'=\langle YZY, XZZ, YXZ \rangle$. To $M_i$ correspond the following Clifford operators that make \eqref{e-Diagram1} commute:
 \[
 U_1=\frac{1}{\sqrt{2}}\left(\!\!\begin{array}{cc}1&i\\i&1\end{array}\!\!\right) 
 ,U_2 =\frac{1}{\sqrt{2}} \left(\!\!\begin{array}{cc}
1 & 1  \\
1 & -1 \end{array} \!\!\right), U_3 = \left(\!\!\begin{array}{cc} 1 & 0 \\ 0 & i \end{array}\!\!\right).
\]
One easily verifies $S'=U\sigma(S)U^\dagger$ where $U=U_1\otimes U_2\otimes U_3$. The corresponding quantum stabilizer states $\cQ(S)$ and $\cQ(S')$ are the one-dimensional complex spaces generated by vectors $v=(1,0,0,-1,0,1,1,0)\T$ and $v'=(1,1,-i,i,1,-1,-i,-i)\T$ respectively. By Theorem \ref{T-LCP1} and Remark \ref{R-USU} we have 
\begin{equation}\label{e-Exx}
\cQ(S') = \cQ(U\sigma(S)U^\dagger) = U\cQ(\sigma(S)).
\end{equation}
Note that $\sigma(S) = \langle ZXX,XXZ,ZZZ\rangle$ and $\cQ(\sigma(S))$ is generated by $v''=(1,0,0,1,0,-1,1,0)\T$. One could also verify \eqref{e-Exx} directly by noting that $Uv$ and $v''$ differ only by the scalar $(1+i)/2$.
\end{exa}
%%%%%%%%%%%%%%%%%%%%%%%%%%%%%%%%%%%%%%%%
\begin{exa}\label{Ex-LCP}
Let $\cC = \im G$ and $\cC' = \im G'$ be the self-dual stabilizer codes where $G$ and $G'$ are as follows  
\[
G=\left(\!\!\begin{array}{cccc|cccc}1&0&1&1&0&1&0&0\\0&1&0&1&1&0&0&0\\0&0&0&0&1&0&1&0\\0&0&0&0&1&1&0&1\end{array}\!\!\right),\quad
G' = \left(\!\!\begin{array}{cccc|cccc}1&1&1&1&1&0&0&1\\0&0&1&1&0&1&0&0\\0&0&0&0&0&0&1&1\\0&0&1&1&1&0&0&0\end{array}\!\!\right).
\]
The map $f:\cC\longrightarrow \cC'$ that maps the $i$-th row of $G$ to the $i$-th row of $G'$ is a symplectic isometry and thus $\cC$ and $\cC'$ are symplectially equivalent. On the other hand, it is easy to see that there cannot exist a $\SL_2(\F_2)$-monomial map between the two. The associated stabilizers are
\begin{equation}\begin{split}S = \langle XZXX, ZXIX, ZIZI, ZZIZ \rangle,\\ S' =\langle YXXY, IZXX, IIZZ, ZIXX \rangle.\end{split}\end{equation}
Then, the respective quantum stabilizer states are
\begin{equation}\label{e-NotLU}
\begin{split}
\cQ(S) = \spann_{\C}\{(1,0,0,0,0,0,0,-1,0,0,1,0,0,1,0,0)\T\},\\
\cQ(S') = \spann_{\C}\{(1,0,0,1,0,0,0,0,0,0,0,0,1,0,0,-1)\T\}.
\end{split}
\end{equation}
Since $f$ is not a $\SL_2(\F_2)$-monomial map Theorem \ref{T-LCP1} implies that $\cQ(S)$ and $\cQ(S')$ are not LCP equivalent. In fact, they are not even LU equivalent. To show this we make use of the \textbf{vectorization} of matrix, that is, $\text{vec}(X)$ of a matrix $X$ is the column vector where we stack the columns of $X$. Let $X, X'\in \cM_4(\F_2)$ be the matrices whose vectorization gives the vectors in \eqref{e-NotLU}. Namely
\begin{equation}
X = \left(\!\!\begin{array}{cccc} 1&0&0&0\\0&0&0&1\\0&0&1&0\\0&-1&0&0\end{array}\!\!\right) \text{ and }
X' = \left(\!\!\begin{array}{cccc}1&0&0&1\\0&0&0&0\\0&0&0&0\\1&0&0&-1\end{array}\!\!\right).
\end{equation}
Assume that there exists $U = U_1\otimes U_2\otimes U_3\otimes U_4\in \cU(2)^{\otimes 4}$ such that $\cQ(S') = U\cQ(S)$. From elementary properties of the Kronecker Product, this is equivalent with
\begin{equation}
(U_3\otimes U_4)X(U_1\T\otimes U_2\T) = X'.
\end{equation}%\MyNote{The same rank is necessary. \\Problem: What unitary produce isometries?}Clearly this is impossible since the right-hand-side has rank 2 whereas the left hand side has rank 4.
Clearly this is impossible since the right-hand-side has rank 2 whereas the left hand side has rank 4.
\end{exa}

\section{Conclusions and Future Research}\label{S-6}
We have presented a detailed study of symplectic isometries of stabilizer codes. For stabilizer codes over fields we establish how far from being a $\SL_2(\F_q)$-monomial map a symplectic isometry is. This is achieved via Theorem \ref{T-main}. However, as discussed in Section \ref{Sec-Symp}, the stabilizer codes constructed with predetermined isometry groups are asymptotically bad. Indeed, the rate goes to zero as the characteristic of the alphabet goes to infinity. 
\begin{problem}
Construct asymptotically good stabilizer codes that satisfy Theorem \ref{T-main}.
\end{problem}
For the general case over local Frobenius rings a partial result is presented. In this case, the group $\Symp(\cC)$ is easily understood and related with the case of stabilizer codes over fields. Whereas, since $\SL_2(R) \neq \SL_2(\F_q)$, the techniques presented in this paper do not help toward understanding $\Mon_{\text{SL}}(\cC)$.
\begin{problem}
Establish an analogous result as in Theorem \ref{T-main} for stabilizer codes over Frobenius rings.
\end{problem}
In Section \ref{Sec-LULC} we relate equivalence notions of quantum stabilizer codes with symplectic isometries. In particular, Theorem \ref{T-LCP1} characterizes LCP equivalence in terms of $\SL_2(R)$-monomial maps. We view this as the first step toward systematically constructing LU equivalent stabilizer states that are not LC. Of course, much more work is needed to understand the structure of counterexamples of LU-LC conjecture. The strategy for searching for such counterexamples was already pointed out in Examples \ref{Ex-LCP}. We make this precise here as the main future direction. Let $\cC=\im G,\,\cC'=\im G'\subseteq \F_q^{2n}$ be two stabilizer codes of the same dimension. Define two isometry groups 
\begin{align*}
   \rMon(\cC,\,\cC')&:=\{B\in \GL_k(\F_q) \mid GM_{|\cC} = BG', M \text{ is an $\SL_2(\F_q)$-monomial map}\},\\
   \Symp(\cC,\,\cC')&:=\{B\in \GL_k(\F_q)\mid \swt(xG)=\swt(xBG') \text{ for all } x\in\F_q^k\}.
   \end{align*}
Example \ref{Ex-LCP} shows that $\rMon(\cC,\,\cC') \subsetneq \Symp(\cC,\,\cC')$ in general. Let $f\in \Symp(\cC,\,\cC') -\rMon(\cC,\,\cC')$. Since $f \notin \rMon(\cC,\,\cC')$, Theorem \ref{T-LCP1} guarantees that $\cQ(\Psi^{-1}(\cC))$ and $\cQ(\Psi^{-1}(\cC'))$ cannot be LCP stabilizer codes. So if they are LU equivalent to start with, we have a counterexample. Unfortunately it is not clear how LU equivalence fits into the language of Section \ref{Sec-Symp}. Thus more work is needed for understanding what symplectic isometries produce LU equivalent quantum stabilizer codes. As far as LU-LC conjecture is concerned we may restrict ourselves on quantum stabilizer states, to which correspond self-dual stabilizer codes.
\begin{problem}
Let $\cC,\,\cC'\subseteq \F_q^{2n}$ be two self-dual stabilizer codes. Establish how different $\rMon(\cC,\,\cC')$ and $\Symp(\cC,\,\cC')$ can be. That is, let $H,\,K\leq \GL_n(\F_q)$ be two groups that satisfy some reasonable necessary conditions. Is it possible to construct two self-dual stabilizer codes $\cC$ and $\cC'$ such that $H=\rMon(\cC,\,\cC')$ and $K=\Symp(\cC,\,\cC')$? 
\end{problem}
\begin{problem}
Let $\cC,\,\cC'\subseteq \F_q^{2n}$ be two self-dual stabilizer codes, and $f:\cC\longrightarrow \cC'$ be a symplectic isometry. Find sufficient conditions for the existence of $U\in \cU(q)^{\otimes n}$ with $\cQ(\Psi^{-1}(\cC'))=U\cQ(\Psi^{-1}(\cC))$.
\end{problem}
Note that a rather weak necessary condition for symplectic isometries that produce LU states was mentioned in Example \ref{Ex-LCP}. Namely, if $v$ and $v'$ are generators of two quantum stabilizer states, then the $n\times n$ matrices $X,\, X'$ with $v = \text{vectorization}(X)$ and $v' = \text{vectorization}(X')$ must have the same rank.
%%%%%%%%%%%%%%%%%%%%%%%%%%%%%%%%%%%%%%%%%%%%%%%%%%
\section*{Acknowledgements}
I am grateful to Jay Wood for his inspiring suggestions on isometry groups. I would also like to thank Heide Gluesing-Luerssen for her continuous support and recommendations that helped with preparation and presentation of this work.

\bibliographystyle{abbrv}
\bibliography{SISCbib}

\end{document}